\documentclass[11pt]{article}

\usepackage[T1]{fontenc}
\usepackage{lmodern,microtype}
\usepackage[margin=0.9in]{geometry}
\usepackage{amsmath,amssymb,amsthm,booktabs,array}
\usepackage[ruled,vlined]{algorithm2e}
\usepackage{placeins}
\usepackage{graphicx}
\usepackage[backend=biber,style=numeric,sorting=none,giveninits=true,maxbibnames=99,doi=true,url=true,eprint=true]{biblatex}
\usepackage[colorlinks=true,allcolors=blue]{hyperref}
\hypersetup{
  pdftitle={GRACE: A General, Exact Changepoint Calculus},
  pdfauthor={Michael A. Grantham},
  pdfkeywords={changepoint type, model selection, functional pruning, dynamic programming, piecewise polynomial, seasonality}
}

\SetAlFnt{\small}
\newcommand{\paperfigure}[2]{%
  \IfFileExists{figures/#1.pdf}{%
    \includegraphics[width=\linewidth,height=0.72\textheight,keepaspectratio]{figures/#1.pdf}%
  }{%
    \IfFileExists{figures/#1.png}{%
      \includegraphics[width=\linewidth,height=0.72\textheight,keepaspectratio]{figures/#1.png}%
    }{%
      \PackageError{GRACE}{Missing figure figures/#1.pdf or figures/#1.png}{Upload the required figure before compiling the submission.}%
    }%
  }%
}

\newtheorem{theorem}{Theorem}
\newtheorem{proposition}[theorem]{Proposition}

\title{GRACE: A General, Exact Changepoint Calculus}
\author{Michael A. Grantham}
\date{}

\begin{document}
\maketitle
\begin{abstract}
We introduce GRACE (the General Regime-Aware Changepoint Estimator), a general, globally exact multiple-changepoint detection and estimation framework in which arbitrary numbers of model parameters may change at or persist across estimated changepoint locations, and all permissible changepoint types and patterns of parameter sharing compete within a single $\ell_0$-penalized optimization. The analyst specifies which types of changepoints are permitted and which parameters each type must preserve, without specifying which changes actually occur, where they occur, or in what order. GRACE jointly estimates the number and locations of changepoints, the type of change at each location, the model fitted between changepoints, and all continuous parameters. As a direct consequence, GRACE can estimate models in which the slope persists across a mean shift.

The generality of the framework is achieved through an exact dynamic-programming recursion that retains the accumulated fitting cost conditional on whichever parameters a subsequent change may preserve, tracking multiple parameters jointly when necessary. We establish the global exactness of the recursion and develop functional, penalty-based, and objective-bound pruning rules that eliminate candidates without sacrificing optimality. Ordinary optimal partitioning, FPOP, and CPOP are recovered as special cases, while piecewise-polynomial and seasonal models illustrate the broader range of permissible segment structures and transitions. The resulting framework supports heterogeneous changepoint types, parameter sharing across segments, and substantial parallel computation within a unified approach to exact changepoint estimation.
\end{abstract}

\noindent\textbf{Keywords:} changepoint detection; changepoint type selection; parameter sharing; exact dynamic programming; functional pruning; $\ell_0$-penalized model selection; piecewise-polynomial regression; seasonality.

\section{Introduction}

Many time series are well described locally by simple curves, although the relationship between successive curves may change from one boundary to the next. A signal may remain continuous while its slope changes, jump to a new level while continuing along the same trend, become exactly constant, resume growth after a plateau, or begin to exhibit curvature. Ordinarily, the analyst chooses among these possibilities before fitting the model and the resulting changepoint algorithm then searches for locations under one prescribed definition of what may change.

GRACE provides a unified exact method for this broader class of problems. The analyst declares the permitted segment forms and boundary relationships, after which one conditional-cost recursion jointly selects the changepoint locations, the type of transition at each boundary, the form fitted between boundaries, and all continuous parameters. The transition dictionary determines which scalar or joint conditional states the recursion must retain. For finite dictionaries of linear-basis least-squares models subject to affine equality constraints, extension and profiling preserve quadratic candidate costs (Proposition~\ref{prop:quadratic-closure}), and the resulting recursion attains the global minimum under the stated sufficiency conditions (Theorem~\ref{thm:exact}). Thus additional segment forms and boundary relationships enter the same exact construction through their declared losses and inherited quantities. Ordinary optimal partitioning is recovered when no continuous quantity is retained. FPOP is recovered for piecewise-constant segments by exposing the current segment level while allowing each changepoint to reset it \cite{fpop}; CPOP is recovered when every segment is linear and every boundary preserves endpoint level \cite{cpop}. Richer procedures arise by exposing whichever scalar or joint vector a permitted transition preserves.

Graph-constrained methods such as gfpop encode permitted movements among discrete regimes while retaining a scalar continuous parameter \cite{gfpop}. Geometric functional pruning and DUST treat multivariate segment parameters when successive segments are otherwise independent \cite{geom,dust}. GRACE permits different transitions to inherit different projections of the preceding fit and retains those quantities jointly whenever a later segment depends on them together. This transition-dependent inheritance is the central object of the calculus and allows the same exact recursion to encompass a much larger collection of changepoint problems.

Our objective is an $\ell_0$-penalized model-selection problem in which the number, locations, and types of changepoints are all unknown. Let
\[
0=\tau_0<\tau_1<\cdots<\tau_m<\tau_{m+1}=n
\]
denote the segment boundaries. Let $a_j$ identify the form fitted on segment $j$, let $\theta_j$ contain its continuous parameters, and let $d_j$ identify the type of transition at $\tau_j$. We call
\[
\mathcal S=(m,\tau,a,d)
\]
a \emph{changepoint story}: it specifies the number and locations of the changepoints, the form fitted on each segment, and the relationship imposed at every boundary. The set $\Theta(\mathcal S)$ contains the continuous parameter values satisfying those relationships.

Writing $C_{\tau_{j-1},\tau_j}^{a_j}(\theta_j)$ for the loss on segment $j$, the desired minimization is
\begin{equation}
(\widehat{\mathcal S},\widehat{\boldsymbol{\theta}})
\in
\arg\min_{\substack{\mathcal S\\ \boldsymbol{\theta}\in\Theta(\mathcal S)}}
\left\{
\sum_{j=1}^{m+1} C_{\tau_{j-1},\tau_j}^{a_j}(\theta_j)
+
\beta_{\mathrm{initial},a_1}
+
\sum_{j=1}^{m}\beta_{d_j}
\right\}.
\label{eq:typed-objective}
\end{equation}

Each selected boundary contributes a discrete penalty, whose size may depend
on the transition type and the number of parameters introduced there.
The resulting $\ell_0$-penalized optimization determines not only where
changepoints occur, but also what changes at each boundary and which
parameters persist across it.

CPOP \cite{cpop} solves an important special case of this problem.
It fits continuous piecewise-linear models, allowing the slope to change
at estimated changepoints while requiring adjacent segments to share
their fitted boundary value. This continuity restriction excludes
discontinuous changes in level, including shifts that preserve the
underlying slope and complete resets of the signal. Moreover, although
a horizontal segment belongs to CPOP's model class, its unrestricted
linear fits do not generally produce exactly zero slopes in the presence
of continuous noise. GRACE permits these different structural
relationships to compete within the same globally exact optimization.

The range of models available even in this first-order setting is
considerable. Figure~\ref{fig:grace-joint-recovery} illustrates a
single series containing six qualitatively different changepoint types:
trend termination, trend resumption, a discontinuous level shift
preserving slope, a continuous slope change, a complete reset of level
and slope, and a reset to a constant regime. GRACE identifies all six transition types in their generating order, recovers four changepoint locations exactly, and estimates each of the remaining two locations within one observation. Its fitted constant regimes are exactly horizontal, and
the segments adjoining the level shift share an estimated slope.

\begin{figure}[!tbp]
\centering
{\large\bfseries
Joint recovery of heterogeneous changepoint types\par}
\vspace{0.5em}
\includegraphics[
    width=\linewidth,
    height=.56\textheight,
    keepaspectratio
]{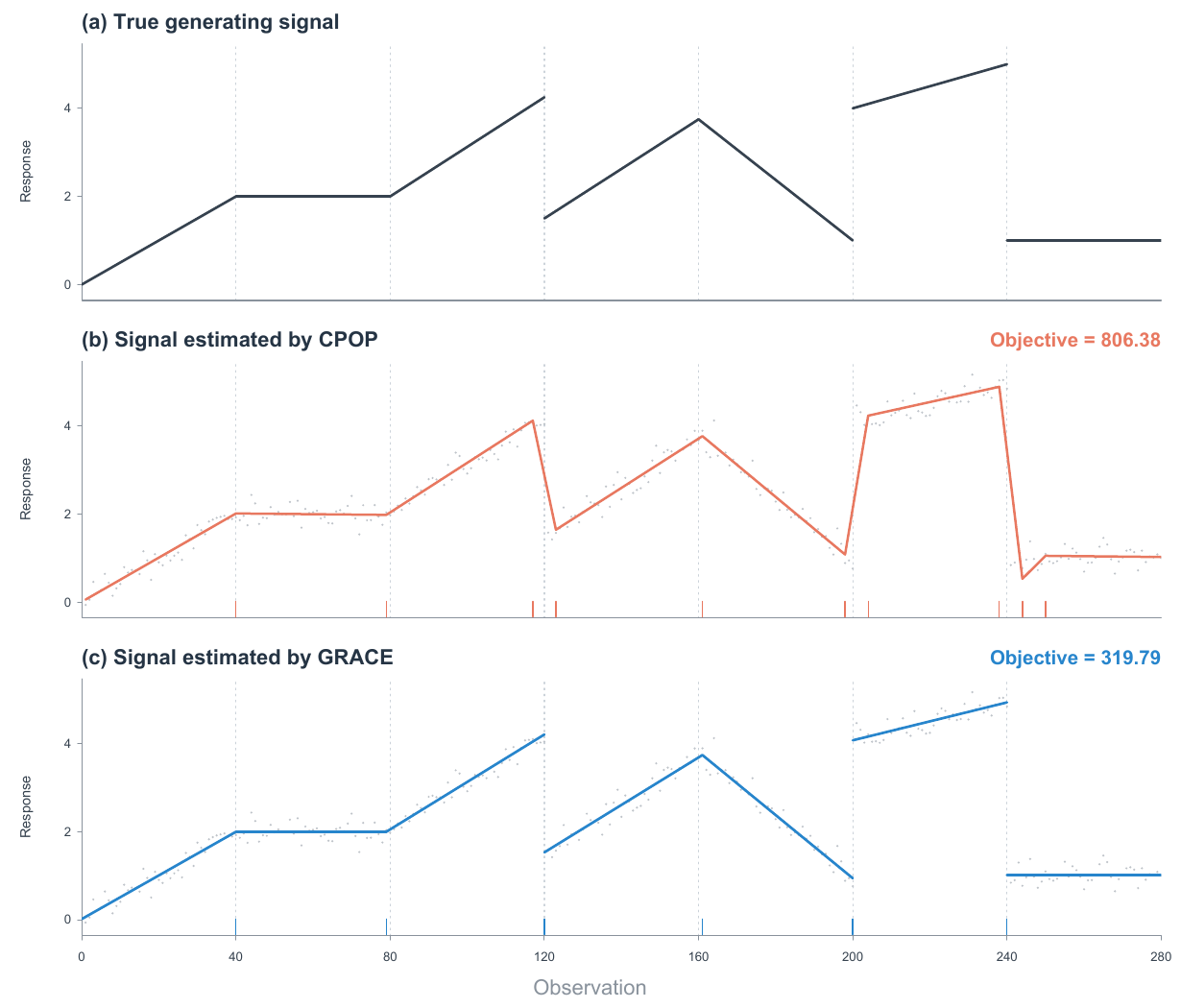}
\caption{Joint recovery of heterogeneous first-order transitions.
(a) The true generating signal. (b) The signal estimated by CPOP and (c) the signal estimated by GRACE, using identical observations and a minimum segment length of six for both methods.
Truth is charcoal, CPOP estimates are orange, and GRACE estimates
are blue. Dotted vertical guides mark generating changepoints,
and colored ticks mark estimated changepoints. GRACE recovers
all six generating transition types, including exactly constant
regimes and a discontinuous level shift preserving slope.
CPOP instead introduces additional continuous slope changes to
approximate discontinuities, while its nearly horizontal segments
are not constrained to have zero slope. GRACE recovers the
generating structure with fewer changepoints and a smaller
penalized objective. The observations contain independent
Gaussian noise with standard deviation $0.2$. In GRACE's optimization, all changepoint locations, transition types, and continuous parameters are selected jointly within a single globally exact optimization. The displayed penalized objectives use \(\mathrm{RSS}/\sigma^2\) as the loss.}
\label{fig:grace-joint-recovery}
\end{figure}

CPOP, constrained to continuous linear segments with minimum segment length of 6, instead selects
ten changepoints, introducing steep connecting segments to
approximate discontinuities and estimating nearly horizontal lines
where the generating signal is constant. Under the stated penalty
convention and common scaled loss, GRACE achieves a penalized
objective of 319.79, compared with 806.38 for CPOP. 

The example illustrates the central advantage of allowing heterogeneous
changepoint types within a single optimization. The analyst specifies
which structural changes are scientifically permissible, but need
not determine in advance which types actually occur, how frequently
they occur, or in what order. GRACE selects these relationships
jointly with the changepoint locations, segment forms, and continuous
parameters. The remainder of the paper develops the exact mathematical
construction underlying this generality, beginning with the
continuous piecewise-linear setting.

We first give an exact algorithm for estimating this richer linear model under the $\ell_0$ objective in \eqref{eq:linear-objective}. We then generalize the same construction to a broad class of models, including polynomial regimes with persistent seasonality.

\section{The CPOP engine and its level-conditioned extensions}
\label{sec:cpop-engine}

\subsection{A shared boundary value}

Suppose that observations $y_1,\ldots,y_n$ occur at locations $1,\ldots,n$. Let
\[
0=\tau_0<\tau_1<\cdots<\tau_m<\tau_{m+1}=n
\]
denote the segment boundaries, and let $\phi_j$ be the fitted value at boundary $\tau_j$. Between two successive boundaries $r$ and $t$, a continuous linear segment starts at value $\psi$ and ends at value $\phi$. Its residual sum of squares is
\begin{equation}
C_{r,t}(\psi,\phi)
=
\sum_{i=r+1}^{t}
\left\{
y_i-\psi-\frac{i-r}{t-r}(\phi-\psi)
\right\}^{2}.
\label{eq:segment-cost}
\end{equation}

CPOP minimizes
\begin{equation}
\min_{m,\tau,\phi}
\left\{
\sum_{j=1}^{m+1}
C_{\tau_{j-1},\tau_j}(\phi_{j-1},\phi_j)
+
m\beta_{\mathrm{CPOP}}
\right\},
\label{eq:cpop-objective}
\end{equation}
subject to the chosen minimum segment length and any other admissibility conditions \cite{cpop}. The shared value $\phi_j$ appears in the segment cost on each side of the $j$th changepoint. Estimating the two segments independently would generally produce different values at the boundary and would no longer solve the continuous model.

\subsection{The conditional-cost recursion}

Let $F_t^\phi(\phi)$ be the minimum cost through observation $t$ when the fitted value at $t$ is held at $\phi$ and the superscript $\phi$ indicates that the envelope is partitioned over \emph{level}. Earlier changepoints and fitted values are optimized, but the final value remains an argument of the cost function.

If the last changepoint before $t$ occurs at $r$, write $\psi$ for the fitted value shared by the earlier fit and the final segment. Conditional on $\psi$, the best earlier fit costs $F_r^\phi(\psi)$. Appending the final segment and then minimizing over the shared value gives
\begin{equation}
F_t^\phi(\phi)
=
\min_{0\leq r<t}
\min_{\psi}
\left\{
F_r^\phi(\psi)
+
C_{r,t}(\psi,\phi)
+
\beta_{\mathrm{CPOP}}
\right\}.
\label{eq:cpop-recursion}
\end{equation}
Setting $F_0^\phi(\psi)=-\beta_{\mathrm{CPOP}}$ makes the $r=0$ alternative a one-segment fit with no changepoint penalty.

The earlier cost $F_r^\phi(\psi)$ and the new segment cost $C_{r,t}(\psi,\phi)$ are added before $\psi$ is minimized, allowing observations on both sides of the boundary to determine their common fitted value. The crucial algebraic fact is that this operation preserves the form of the candidate costs. Under least squares, each candidate represented in $F_r^\phi(\psi)$ is quadratic in $\psi$, and after the segment cost is added, the resulting expression is jointly quadratic in $(\psi,\phi)$. Profiling over $\psi$ is therefore available in closed form and leaves another quadratic, now in $\phi$. This closure is the small miracle behind CPOP: after a segment is appended and its starting value is jointly estimated, the accumulated cost retains exactly the functional form needed to repeat the recursion at the next boundary.

Every possible changepoint history reaching $t$ consequently contributes a quadratic function of $\phi$, and $F_t^\phi(\phi)$ is their pointwise minimum. A candidate cannot be discarded merely because its own minimum is comparatively large, since it may still be the cheapest candidate for some range of endpoint values and may therefore lead to the eventual optimum. Functional pruning removes it only when it fails to attain the lower envelope for every $\phi$.

\subsection{What can be done with the level envelope alone}

The quadratic closure underlying CPOP survives several useful extensions of the permitted model. Trend termination, trend resumption, and complete resets each add a least-squares segment cost and profile over any parameters that are not carried forward to the new endpoint. Because these operations preserve quadratic form, every resulting candidate is available in closed form as a quadratic in the new endpoint value $\phi$. The endpoint level therefore remains the only continuous information passed to the next boundary, and all of these candidates can be incorporated directly into CPOP's level-conditioned envelope machinery. We write $F_{t,\mathrm{lin}}^\phi$ and $F_{t,\mathrm{const}}^\phi$ for the level-conditioned envelopes restricted to fits whose current segment is respectively linear or constant.

Suppose first that an incoming linear regime is followed by a constant segment. Preserving the boundary value while fixing the outgoing slope at zero gives
\[
C_{r,t}(\phi,\phi)=\sum_{i=r+1}^{t}(y_i-\phi)^2.
\]
Each incoming linear candidate is already quadratic in its endpoint value $\phi$, and the constant-segment cost is also quadratic in $\phi$. Their sum is therefore another quadratic in the same endpoint value. At the envelope level, the update is
\[
F_{r,\mathrm{lin}}^\phi(\phi)+C_{r,t}(\phi,\phi).
\]
This transition, which we call \emph{trend termination}, jointly estimates the final value of the incoming line and the level of the outgoing constant regime. The incoming slope remains freely estimated from the preceding observations; only the outgoing slope is fixed at zero.

Conversely, suppose that a constant regime is followed by an unrestricted line that inherits its level. If the new line begins at $\psi$ and ends at $\phi$, the update is
\[
\min_{\psi}\left\{F_{r,\mathrm{const}}^\phi(\psi)+C_{r,t}(\psi,\phi)\right\}.
\]
For each candidate in the incoming constant envelope, the expression inside braces is jointly quadratic in $(\psi,\phi)$. Profiling over the shared value $\psi$ is therefore available in closed form and leaves a quadratic in $\phi$, exactly as in the original CPOP recursion. We call this transition \emph{trend resumption}. Its outgoing slope is estimated jointly with the inherited boundary value and is permitted, but not required, to be nonzero.

A complete reset requires no continuous information from the incoming fit. Define the unconditional minimum through $r$
\begin{equation}
F_r^\varnothing=\min_{\phi}F_r^\phi(\phi),
\label{eq:unconditional-linear}
\end{equation}
where the minimization also ranges over compatible constant and linear fits through $r$ when both are permitted. For a reset to an unrelated line, the candidate passed to the level envelope is
\[
F_r^\varnothing+\min_{\psi}C_{r,t}(\psi,\phi).
\]
Because the segment cost is jointly quadratic in $(\psi,\phi)$, profiling over its freely estimated starting value $\psi$ again yields a closed-form quadratic in $\phi$. A reset to an unrelated constant produces
\[
F_r^\varnothing+C_{r,t}(\phi,\phi),
\]
which is already quadratic in $\phi$.

Thus trend termination, trend resumption, and complete resets all generate candidate costs of the same form required by CPOP: closed-form quadratics conditional only on the current endpoint value $\phi$. Transition penalties add constants and therefore do not disturb this closure. These candidates can consequently be folded into the level-conditioned envelope and propagated by the existing recursion without introducing another continuous state. The transitions need not receive the same penalty because they introduce different numbers of continuous parameters: trend termination fixes the outgoing slope, trend resumption estimates a new slope, and a complete reset estimates a new level and possibly a new slope. The precise penalties are defined alongside the transition dictionary in the next section.

\subsection{The transition the level envelope cannot supply}

Now consider a discontinuous level shift that preserves the rate of change. The fitted values on the two sides of the boundary are unrelated, but the incoming and outgoing slopes must agree. Observations on both sides of the jump must therefore contribute to the estimation of one common slope.

For any generated unrestricted linear candidate on $(r,t]$, the slope of its final segment can be recovered directly from its starting and ending values $\psi$ and $\phi$. Each generated candidate therefore contains the information needed to condition on slope. CPOP, however, organizes and prunes these candidates according to their endpoint level: it retains those candidates that attain the lower envelope for some value of $\phi$. The resulting regions of optimality are regions over level, and they need not correspond to regions of optimality over slope. A candidate that never attains the level envelope may nevertheless be the cheapest candidate at a particular slope. If it is discarded on the basis of level dominance, no transformation of the surviving envelope can recover that slope-conditioned optimum.

A simple algebraic example illustrates the information lost by pruning over level alone. Let $\phi$ denote endpoint level and $\dot\phi$ slope, and consider two quadratic candidate costs with the same permitted continuations:
\[
Q_1(\phi,\dot\phi)=\phi^2+\dot\phi^2,
\qquad
Q_2(\phi,\dot\phi)=\phi^2+(\dot\phi-2)^2+1.
\]
Their level-conditioned profiles are
\[
\min_{\dot\phi}Q_1(\phi,\dot\phi)=\phi^2,
\qquad
\min_{\dot\phi}Q_2(\phi,\dot\phi)=\phi^2+1.
\]
The second candidate therefore loses at every endpoint level. Their slope-conditioned profiles, however, are
\[
\min_{\phi}Q_1(\phi,\dot\phi)=\dot\phi^2,
\qquad
\min_{\phi}Q_2(\phi,\dot\phi)=(\dot\phi-2)^2+1.
\]
At slope $2$, the second candidate costs $1$ while the first costs $4$. Adding the same future cost $(\dot\phi-2)^2$ gives minimized totals of $2$ and $1$, respectively. Thus a candidate discarded everywhere by the level envelope can supply the better completed fit when the continuation depends on slope.

The remedy is to construct both envelopes before pruning either one. For a segment $(r,t]$, endpoint level and slope are related by
\begin{equation}
\dot\phi=\frac{\phi-\psi}{t-r},
\qquad
\psi=\phi-(t-r)\dot\phi.
\label{eq:slope-transform}
\end{equation}
GRACE uses this relationship to profile every generated candidate in two ways: once at fixed endpoint level and once at fixed slope. Both profiles are available in closed form as quadratics in the retained quantity. They are then placed into separate level- and slope-conditioned envelopes, each of which is pruned according to its own regions of optimality. A candidate eliminated from one envelope consequently remains available if it is required by the other.

\section{Exact optimization over linear changepoint types}
\label{sec:linear-grace}

\subsection{Two conditional views of each generated candidate}

Consider an unrestricted line on $(r,t]$ appended to the familiar level-conditioned cost through $r$. Apart from its additive transition penalty, the accumulated cost is formed from
\[
F_r^\phi(\psi)+C_{r,t}(\psi,\phi).
\]
The following operations are applied separately to each retained quadratic represented in $F_r^\phi$ before the newly generated candidates are compared at $t$.

Holding the new endpoint value fixed and profiling over the shared value $\psi$ gives
\begin{equation}
\min_\psi\left\{F_r^\phi(\psi)+C_{r,t}(\psi,\phi)\right\}.
\label{eq:level-profile}
\end{equation}
This is the ordinary CPOP update. It produces the candidate quadratics whose pointwise minimum forms the new level-conditioned envelope $F_t^\phi(\phi)$.

The same generated candidates can instead be viewed conditional on the slope of their final segment. Substituting $\psi=\phi-(t-r)\dot\phi$ and profiling over the remaining endpoint value gives
\begin{equation}
\min_\phi\left\{F_r^\phi\bigl(\phi-(t-r)\dot\phi\bigr)+C_{r,t}\bigl(\phi-(t-r)\dot\phi,\phi\bigr)\right\}.
\label{eq:slope-profile}
\end{equation}
Under least squares, the accumulated cost before either profiling operation is jointly quadratic in $(\psi,\phi)$. The level profile is therefore available in closed form as a quadratic in $\phi$, while the slope profile is available in closed form as a quadratic in $\dot\phi$. The pointwise minimum of the slope-conditioned quadratics forms the second envelope $F_t^{\dot\phi}(\dot\phi)$. Appendix~\ref{app:transformation} gives the explicit coefficient transformation.

In the first-order model, $F_t^{\dot\phi}$ ranges over candidates whose current regime is an unrestricted line, because only such a regime supplies an incoming slope for a later slope-preserving transition.

Other transition types change the accumulated cost through $r$ that is combined with $C_{r,t}$, but the two profiling operations remain the same. Every generated unrestricted linear candidate is placed into both conditional collections before either is pruned. The candidates attaining the two envelopes are then retained independently: a candidate may contribute to the slope envelope even if it contributes nowhere to the level envelope, or conversely.

The level-conditioned envelope therefore remains sufficient for continuous slope changes, trend termination, and trend resumption, while its unconditional minimum supplies complete resets. The slope-conditioned envelope adds the one missing possibility: a discontinuous level shift that preserves the incoming slope. Together, these operations generate all the boundary relationships used in the first-order model, which we now collect according to what each transition preserves, estimates, or fixes. 

\subsection{The transition dictionary and objective}

Table~\ref{tab:linear-dictionary} collects the permitted relationships between successive constant and linear regimes. We call this collection the \emph{changepoint dictionary}. Each entry specifies which quantities are inherited from the preceding regime, which are estimated anew, and which are fixed in the outgoing regime. Choosing an entry at each changepoint, together with the changepoint locations and segment forms, determines the changepoint story introduced in Section~1.

\begin{table}[ht]
\centering
\small
\begin{tabular}{p{0.29\linewidth}p{0.20\linewidth}p{0.23\linewidth}p{0.18\linewidth}}
\toprule
Transition & Preserved & Newly estimated & \shortstack{Fixed outgoing\\quantity}\\
\midrule
Continuous slope change & Level & Slope & None\\
Level shift preserving slope & Slope & Level & None\\
Complete linear reset & None & Level and slope & None\\
Trend termination & Level & None & Slope $=0$\\
Trend resumption & Level & Slope & None\\
Reset to a constant & None & Level & Slope $=0$\\
\bottomrule
\end{tabular}
\caption{The first-order transition dictionary for constant and linear regimes.}
\label{tab:linear-dictionary}
\end{table}

Give segment $j$ its own starting and ending values $\psi_j$ and $\phi_j$, let $a_j\in\{\mathrm{const},\mathrm{lin}\}$ record its form, and let $d_j$ identify the transition at $\tau_j$. GRACE minimizes
\begin{equation}
\min_{\substack{m,\tau,a,d,\psi,\phi\\\mathrm{admissible}}}
\left\{
\sum_{j=1}^{m+1}
C_{\tau_{j-1},\tau_j}(\psi_j,\phi_j)
+
\beta_{\mathrm{initial},a_1}
+
\sum_{j=1}^{m}\beta_{d_j}
\right\},
\label{eq:linear-objective}
\end{equation}
subject to the boundary relationships selected by the story. A continuous transition equates $\phi_j$ with $\psi_{j+1}$. A level shift preserving slope equates the slopes of the neighboring segments but imposes no equality between their boundary levels. A reset imposes no relationship between the incoming and outgoing continuous parameters. A constant segment imposes $\psi_j=\phi_j$.

Following CPOP \cite{cpop}, we use the following Schwarz-motivated criterion \cite{schwarz}. Let
\[
\beta=\sigma^2\log n.
\]
If transition $d$ introduces $p_d$ continuous parameters, its penalty is
\begin{equation}
\beta_d=(1+p_d)\beta,
\label{eq:type-penalty}
\end{equation}
where the first unit accounts for the unknown changepoint location. A parameter preserved from the preceding regime was counted when first introduced and incurs no additional penalty, while fixing an outgoing quantity at a specified value introduces no parameter.

Under this criterion, a continuous slope change, trend resumption, or level shift preserving slope receives penalty $2\beta$; a complete reset to a line receives $3\beta$; trend termination receives $\beta$; and a reset to a constant receives $2\beta$. An initial line receives penalty $2\beta$, while an initial constant receives $\beta$, since no changepoint location is introduced. Appendix~\ref{app:penalty-convention} derives these penalties and explains the treatment of the unknown changepoint locations.

\subsection{The first-order recursion}

For a proposed outgoing line on $(r,t]$, write
\[
\psi=\phi-(t-r)\dot\phi.
\]
The incoming cost depends on the transition type.

A continuous slope change preserves the boundary level and estimates a new slope. Trend resumption does the same when the preceding regime is constant. A level shift preserving slope inherits slope instead of level. A complete reset inherits neither. Combining these alternatives gives
\begin{equation}
C_{r,t}(\psi,\phi)
+
\min\left\{
\begin{aligned}
&F_{r,\mathrm{lin}}^\phi(\psi)+2\beta,
&&\text{continuous slope change},\\
&F_{r,\mathrm{const}}^\phi(\psi)+2\beta,
&&\text{trend resumption},\\
&F_r^{\dot\phi}(\dot\phi)+2\beta,
&&\text{level shift preserving slope},\\
&F_r^\varnothing+3\beta,
&&\text{complete reset to a line}.
\end{aligned}
\right.
\label{eq:linear-update}
\end{equation}
Each term includes only earlier fits for which the named transition is admissible.

For an outgoing constant segment, the corresponding update is
\begin{equation}
C_{r,t}(\phi,\phi)
+
\min\left\{
\begin{aligned}
&F_{r,\mathrm{lin}}^\phi(\phi)+\beta,
&&\text{trend termination},\\
&F_r^\varnothing+2\beta,
&&\text{reset to a constant}.
\end{aligned}
\right.
\label{eq:constant-update}
\end{equation}
A redundant boundary between two unchanged constant regimes need not be included.

Each branch in \eqref{eq:linear-update} generates an accumulated quadratic in the outgoing line parameters. Every such quadratic is then profiled once at fixed endpoint level and once at fixed slope, using \eqref{eq:level-profile} and \eqref{eq:slope-profile}. The profiles are placed in their respective candidate collections before either collection is pruned. Constant candidates contribute to the level envelope and have slope fixed at zero.

Initialization considers every allowed one-segment model. An initial line contributes
\[
C_{0,t}\bigl(\phi-t\dot\phi,\phi\bigr)+2\beta,
\]
and an initial constant contributes
\[
C_{0,t}(\phi,\phi)+\beta.
\]
The terminal optimum is obtained from the unconditional minimum over all admissible candidates ending at $n$. Given here is the GRACE recursion applied to constant and linear regimes.

\begin{algorithm}[H]
\caption{Exact GRACE recursion for constant and linear regimes}
\label{alg:linear-grace}
\KwIn{Observations; admissible segment lengths; permitted transition types; penalties.}
Create separate candidate collections for constant and linear regimes\;
\For{$t=1,\ldots,n$}{
  Add every admissible one-segment constant or linear model ending at $t$\;
  \For{each admissible previous boundary $r<t$}{
    Generate continuous slope-change candidates from $F_{r,\mathrm{lin}}^\phi$\;
    Generate trend-resumption candidates from $F_{r,\mathrm{const}}^\phi$\;
    Generate slope-preserving level-shift candidates from $F_r^{\dot\phi}$\;
    Generate complete linear resets from $F_r^\varnothing$\;
    Generate trend-termination candidates from $F_{r,\mathrm{lin}}^\phi$\;
    Generate constant resets from $F_r^\varnothing$\;
  }
  \For{each generated linear candidate}{
    Profile it at fixed endpoint level and add the result to the level collection\;
    Profile it at fixed slope and add the result to the slope collection\;
    Store the corresponding conditional minimizers\;
  }
  Add generated constant candidates to their level-conditioned collection\;
  Construct the unconditional minima needed by resets and terminal evaluation\;
  Prune the completed envelopes as described in Section~\ref{sec:linear-pruning}\;
}
Minimize over all admissible candidates ending at $n$ and reconstruct the selected story from the stored minimizing choices\;
\KwOut{The optimal objective, changepoint locations and types, and jointly fitted parameters.}
\end{algorithm}

The recursion differs from CPOP in two respects. First, it generates quadratics corresponding to the different permitted changepoint types rather than only to continuous slope changes. Second, every generated linear candidate is profiled both by endpoint level and by slope, and the resulting envelopes are pruned separately. A candidate may therefore survive in one envelope even if it contributes nowhere to the other.

\begin{proposition}[Exactness of the first-order recursion]
The first-order recursion attains the global minimum of \eqref{eq:linear-objective}.
\end{proposition}

\begin{proof}
Separate the two possible one-segment fits from fits containing at least one changepoint, and then isolate the final segment:
\begin{align}
F_n^\varnothing
&=
\min\Bigg\{
\min_{\psi,\phi}\left\{C_{0,n}(\psi,\phi)+2\beta\right\},
\min_{\phi}\left\{C_{0,n}(\phi,\phi)+\beta\right\},
\notag\\
&\qquad
\min_{\substack{m\geq 1,\ \tau,a,d,\psi,\phi\\ \mathrm{admissible}}}
\Bigg[
\sum_{j=1}^{m}
C_{\tau_{j-1},\tau_j}(\psi_j,\phi_j)
+
\beta_{\mathrm{initial},a_1}
+
\sum_{j=1}^{m-1}\beta_{d_j}
+
C_{\tau_m,n}(\psi_{m+1},\phi_{m+1})
+
\beta_{d_m}
\Bigg]
\Bigg\}.
\label{eq:exactness-full-objective}
\end{align}
Writing $r=\tau_m$, $\psi=\psi_{m+1}$, and $\phi=\phi_{m+1}$ gives
\begin{align}
F_n^\varnothing
&=
\min\Bigg\{
\min_{\psi,\phi}\left\{C_{0,n}(\psi,\phi)+2\beta\right\},
\min_{\phi}\left\{C_{0,n}(\phi,\phi)+\beta\right\},
\notag\\
&\qquad
\min_{\substack{0<r<n\\ \psi,\phi}}
\Bigg[
C_{r,n}(\psi,\phi)
+
\min_{d_m}
\Bigg\{
\min_{\substack{m,\tau,a_{1:m},d_{1:m-1}\\ \psi_{1:m},\phi_{1:m}\\ \tau_m=r\\ \mathrm{admissible\ and}\\ \mathrm{compatible\ with}\ d_m}}
\left[
\sum_{j=1}^{m}
C_{\tau_{j-1},\tau_j}(\psi_j,\phi_j)
+
\beta_{\mathrm{initial},a_1}
+
\sum_{j=1}^{m-1}\beta_{d_j}
\right]
+
\beta_{d_m}
\Bigg\}
\Bigg]
\Bigg\}.
\label{eq:exactness-last-transition}
\end{align}
For fixed $r$, $\psi$, $\phi$, and $d_m$, the innermost minimum is the conditional cost through $r$ required by the final transition. The final segment must be either a line or a constant, and the final transition must take one of the forms in Table~\ref{tab:linear-dictionary}. Therefore
\begin{align}
F_n^\varnothing
=
\min\Bigg\{
&\min_{\psi,\phi}\left\{C_{0,n}(\psi,\phi)+2\beta\right\},
\min_{\phi}\left\{C_{0,n}(\phi,\phi)+\beta\right\},
\notag\\
&\min_{0<r<n}
\Bigg\{
\min_{\psi,\phi}
\Bigg[
C_{r,n}(\psi,\phi)
+
\min\left\{
\begin{aligned}
&F_{r,\mathrm{lin}}^\phi(\psi)+2\beta,\\
&F_{r,\mathrm{const}}^\phi(\psi)+2\beta,\\
&F_r^{\dot\phi}\left(\frac{\phi-\psi}{n-r}\right)+2\beta,\\
&F_r^\varnothing+3\beta
\end{aligned}
\right\}
\Bigg],
\notag\\
&\qquad\qquad
\min_{\phi}
\Bigg[
C_{r,n}(\phi,\phi)
+
\min\left\{
\begin{aligned}
&F_{r,\mathrm{lin}}^\phi(\phi)+\beta,\\
&F_r^\varnothing+2\beta
\end{aligned}
\right\}
\Bigg]
\Bigg\}
\Bigg\}.
\label{eq:exactness-final-decomposition}
\end{align}
The four branches in the first inner minimum are, respectively, a continuous slope change, trend resumption, a level shift preserving slope, and a complete reset to a line. The two branches in the second are trend termination and a reset to a constant.

The one-segment terms initialize the recursion. Applying the same decomposition to each conditional cost at $r$ proves by induction that every admissible story appears in \eqref{eq:exactness-final-decomposition} with its objective value, and every term in \eqref{eq:exactness-final-decomposition} corresponds to an admissible story. Hence the displayed minimum equals \eqref{eq:linear-objective}.
\end{proof}

\section{Pruning the linear recursion}
\label{sec:linear-pruning}

The recursion above is exact, but it generates a new collection of quadratics for every admissible previous boundary and transition type. Without pruning, many distinct candidate stories would be carried forward even though they can no longer participate in an optimal solution, and therefore they generate a vast amount of redundant computation. The essential insight behind CPOP is that such candidates can be compared conditionally. If two preceding fits permit the same outgoing segment and one is no more expensive at every possible shared boundary value, the more expensive fit can never become optimal later: the cheaper fit can replace it without changing the outgoing segment or anything that follows.

GRACE uses the same replacement principle, but the quantity passed from one segment to the next depends on the transition. A continuation may require the preceding fit conditional on boundary level, conditional on slope, or only through its unconditional minimum. Pruning must therefore answer each of these conditional comparisons separately. Functional pruning removes candidates that attain none of the relevant envelopes, while the subsequent pruning rules use transition penalties and bounds on the completed objective to eliminate additional candidates.

\subsection{Functional pruning}

In CPOP, each candidate contributes a quadratic function of the fitted endpoint value $\phi$. The lower envelope records the cheapest accumulated cost at every possible value of $\phi$. A candidate is retained precisely when it attains this envelope for at least one attainable value of $\phi$. If it is more expensive than another compatible candidate everywhere, then any future segment beginning at a value it could supply can instead be attached to the cheaper candidate at the same value. Removing it therefore cannot change the final optimum. This is the central functional-pruning insight of CPOP \cite{cpop}: candidates are judged by conditional optimality over the entire endpoint domain, rather than by their individually minimized costs.

The same argument applies to each GRACE envelope. Within the level-conditioned envelope, candidates that support the same future continuations are compared as functions of $\phi$. Within the slope-conditioned envelope, they are compared as functions of $\dot\phi$. A candidate is removed from an envelope only if it fails to attain that envelope anywhere on its attainable domain. Where candidates tie, at least one candidate capable of supporting the same continuations is retained.

\subsection{\texorpdfstring{$k\beta$}{k-beta} pruning}

The functional-pruning principle used above underlies both Rigaill's pruned dynamic-programming algorithm and CPOP \cite{psn,cpop}. CPOP also employs inequality-based pruning: a candidate whose accumulated cost is sufficiently far above a competing cost can be discarded even when it has not yet been removed by functional pruning. GRACE modifies this second idea to account for the different penalties and boundary relationships associated with different transition types.

For each possible next transition $d_j$ in \eqref{eq:linear-objective}, the recursion collects its legal uses across all retained boundary times $r$. We call this collection a \emph{transition family}. The distinction is necessary because the same time may remain useful for a transition that preserves level after it has become uncompetitive for one that preserves slope, or conversely.

Within each transition family, the recursion minimizes over the earlier boundary $r$. The comparison is conditional on the quantity preserved by the transition: endpoint level for a continuous slope change, slope for a level shift preserving slope, and no continuous quantity for a reset. A particular time may be removed from the family only if it loses throughout the corresponding conditional domain.

Inequality-based pruning strengthens this comparison by accounting for the penalty attached to the next transition. Fix a time $r$. Suppose that every outgoing segment obtainable from $r$ through transition $d_j$ can also be obtained through a legal transition $d_j'$, with exactly the same outgoing segment and the same possibilities thereafter. Under \eqref{eq:type-penalty}, the difference between their penalties is an integer multiple of $\beta$. Write
\begin{equation}
\beta_{d_j'}-\beta_{d_j}=k\beta,
\qquad k\geq 0.
\label{eq:kbeta-difference}
\end{equation}
The continuation using $d_j$ may therefore recover a conditional-cost deficit of at most $k\beta$ through its cheaper transition. If, throughout the conditional domain, the accumulated cost through $r$ for $d_j$ exceeds the corresponding cost for $d_j'$ by more than $k\beta$, then adding the transition penalties still leaves $d_j'$ cheaper. Because the two choices produce the same outgoing segment and the same possibilities thereafter, no later observation can reverse that comparison. The following proposition states the resulting pruning rule.

\begin{proposition}[$k\beta$ pruning]
\label{prop:kbeta}
Fix a boundary time $r$ and a permitted transition $d_j$. Suppose there is another permitted transition $d_j'$ such that every outgoing segment available from $r$ through $d_j$ can also be attached through $d_j'$, with the same fitted values on the outgoing segment and the same permitted continuations thereafter. Let $\beta_{d_j'}-\beta_{d_j}=k\beta$ for some $k\geq0$. If, at every attainable value required by $d_j$, the minimum accumulated cost through $r$ compatible with $d_j$ exceeds the minimum accumulated cost needed to attach the same outgoing segment through $d_j'$ by more than $k\beta$, then the use of transition $d_j$ at time $r$ may be removed without changing the global minimum.
\end{proposition}

\begin{proof}
Fix any complete story that uses transition $d_j$ at time $r$ and fix its outgoing segment. By assumption, transition $d_j'$ can attach exactly the same outgoing segment and leaves every later continuation unchanged. Replacing the fit through $r$ and transition $d_j$ by the compatible alternative through $d_j'$ reduces the accumulated cost by more than $k\beta$ and increases the transition penalty by exactly $k\beta$. The replacement therefore has smaller total objective. Since the comparison holds throughout the conditional domain, no complete story using $d_j$ at time $r$ can be optimal.
\end{proof}

The integer $k$ thus determines the amount by which a candidate may trail before it can be pruned from a particular family. Transition families with allowance $k\beta$ will be called \emph{$k$-families}. The case $k=0$ gives ordinary dominance: when two legal continuations have the same penalty, the more expensive one may be removed as soon as it loses everywhere on the relevant conditional domain.

For example, a continuous slope change costs $2\beta$, whereas a complete reset to the same outgoing line costs $3\beta$. The continuous transition has a penalty advantage of $\beta$, so here $k=1$. At a fixed time $r$, the reset can replace the continuous transition whenever
\begin{equation}
F_{r,\mathrm{lin}}^\phi(\psi)+2\beta
>
F_r^\varnothing+3\beta
\qquad\text{for every attainable }\psi.
\label{eq:kbeta-level}
\end{equation}
Both sides then append exactly the same outgoing line; they differ only in the fit through $r$ and in the penalty paid at $r$. Thus \eqref{eq:kbeta-level} removes the continuous-slope-change use of that time. The corresponding slope-preserving use may be removed when
\begin{equation}
F_r^{\dot\phi}(\dot\phi)+2\beta
>
F_r^\varnothing+3\beta
\qquad\text{for every attainable }\dot\phi.
\label{eq:kbeta-slope}
\end{equation}
These are the $k=1$ instances of the general comparison. Larger values of $k$ arise whenever the replacing transition introduces several additional parameters.

Any transition may dominate another when it can produce the same outgoing segment and leaves the same later transitions available. The recursion therefore compares each family with every legal alternative capable of replacing it.

Pruning a time from one transition family does not remove it from the others. A time may be eliminated from a $k=1$ family while remaining in a $k=2$ family, since the latter permits a larger recoverable deficit. The time itself is discarded only after every transition family through which it could still be used has been pruned. Minimum segment lengths and all other admissibility restrictions must be satisfied before one transition can serve as a replacement for another.

\subsection{Upper and lower bounds}

A conditional quadratic is formally defined over every attainable value of the exposed vector $\xi$, but many of those values cannot participate in an optimal story. Here and below, $\xi$ denotes precisely the parameter or joint vector of parameters retained as the argument of a conditional envelope. If fixing $\xi$ already makes the accumulated cost so large that even the most favorable possible completion cannot improve a known complete fit, then that value of $\xi$ is irrelevant. Upper and lower bounds identify these values directly.

Let $U$ be the objective attained by any admissible complete story after all of its continuous parameters have been jointly refitted under the selected boundary relationships. Because this story is feasible, the unknown global optimum is no greater than $U$. For a retained candidate ending at time $t$, write $f_t(\xi)$ for its accumulated conditional cost, where $\xi$ contains the quantities that may be carried into later segments. The lowercase $f_t$ denotes one candidate, whereas $F_t$ continues to denote the lower envelope of all compatible candidates. Let $L_t$ be a lower bound on the additional cost required to fit the remaining observations after $t$, valid for every attainable $\xi$. The candidate can improve the known complete fit at a particular value of $\xi$ only if
\begin{equation}
f_t(\xi)+L_t\leq U.
\label{eq:feasible}
\end{equation}
Values violating \eqref{eq:feasible} cannot occur in any completion with objective at most $U$ and may therefore be excluded from further comparisons.

When observation losses and unpaid penalties are nonnegative, $L_t=0$ is always valid. A stronger lower bound can be obtained by solving an easier version of the problem on the remaining observations $y_{t+1:n}$. Cross-boundary preservation requirements may be removed, future segments may be fitted independently, and their penalties may be reduced, provided that every legal completion remains available in the relaxed problem at no greater cost. The minimum of this relaxed problem cannot exceed the cost of any legal completion and therefore supplies a valid value of $L_t$.

The values satisfying \eqref{eq:feasible} are especially easy to locate for a nondegenerate quadratic candidate. Suppose
\[
f_t(\xi)
=
f_t(\widehat{\xi})
+
(\xi-\widehat{\xi})^\top H(\xi-\widehat{\xi}),
\]
where $\widehat{\xi}$ minimizes the quadratic and $H$ is positive definite. Substituting this expression into \eqref{eq:feasible} gives
\begin{equation}
(\xi-\widehat{\xi})^\top H(\xi-\widehat{\xi})
\leq
U-L_t-f_t(\widehat{\xi}).
\label{eq:ellipsoid}
\end{equation}
If the right-hand side is negative, no value of $\xi$ can improve the incumbent and the candidate may be discarded. If it is zero, only $\widehat{\xi}$ remains relevant. If it is positive, the relevant values form an ellipsoid centered at $\widehat{\xi}$, intersected with the candidate's original attainable domain. In one dimension this ellipsoid is an interval; in two dimensions it is an ellipse. Thus the quadratic itself determines precisely which exposed values remain capable of participating in an optimal story.

If the quadratic is only positive semidefinite, the same inequality remains valid but its solution may be unbounded in directions of zero curvature. Such candidates retain their original feasibility conditions and are discarded only when emptiness or dominance can still be certified.

\begin{proposition}[Objective-bound pruning]
\label{prop:bound}
For a retained candidate ending at time $t$, only values of $\xi$ satisfying \eqref{eq:feasible} can belong to a complete story whose objective is no greater than $U$. If no such value exists, the candidate may be discarded while the complete story attaining $U$ is retained. If compatible candidates have no greater conditional cost throughout the remaining region, the candidate may be removed from that envelope.
\end{proposition}

\begin{proof}
At a retained value $\xi$, every admissible continuation has total objective at least $f_t(\xi)+L_t$. If this exceeds $U$, that value cannot improve the admissible complete story supplying $U$. The candidate may therefore be restricted to the values satisfying \eqref{eq:feasible}; if none remain, it may be discarded. If retained competitors have no greater conditional cost throughout the remaining region, the candidate may be removed from that envelope while preserving a minimizing representative wherever equality occurs.
\end{proof}

\section{The general conditional calculus}
\label{sec:general-calculus}

We now consider models with more continuous parameters: quadratic and cubic trends, quadratic or cubic splines with changepoints as knots, mixtures of linear, quadratic, and cubic pieces, regression models with several covariates, and trends accompanied by seasonal components. The permitted segment forms and boundary relationships are declared in advance. A transition may preserve all, some, or none of the relevant boundary quantities, estimate other outgoing parameters afresh, or fix them at prescribed values. GRACE jointly selects the segment forms, changepoint locations, transition types, and continuous parameters under the objective \eqref{eq:typed-objective}.

The preceding development from CPOP identifies the information this optimization must retain. CPOP exposes endpoint level because a continuous transition inherits that level. A level shift preserving slope requires an additional envelope exposing slope. With richer models, a transition may inherit several quantities together, requiring a joint higher-dimensional envelope. If different transitions preserve different collections of quantities, the recursion must maintain the corresponding collection of envelopes. The same conditional-cost principle therefore extends to richer models once the information required by each transition has been identified.

\subsection{Boundary information and conditional envelopes}

For example, a quadratic segment has a boundary level, slope, and curvature. A curvature change preserving level and slope requires the joint conditional cost
\[
F_t^{\phi,\dot\phi}(\phi,\dot\phi).
\]
A level shift preserving slope and curvature instead requires
\[
F_t^{\dot\phi,\ddot\phi}(\dot\phi,\ddot\phi),
\]
while a slope change preserving level and curvature requires
\[
F_t^{\phi,\ddot\phi}(\phi,\ddot\phi).
\]
A dictionary containing these three transitions therefore requires three distinct two-dimensional envelopes. For a cubic spline whose pieces join with continuous value, first derivative, and second derivative, the corresponding envelope is
\[
F_t^{\phi,\dot\phi,\ddot\phi}
(\phi,\dot\phi,\ddot\phi).
\]
Persistent seasonal or regression coefficients must likewise remain exposed jointly with whichever other quantities a transition preserves.

These joint envelopes cannot generally be reconstructed from separate scalar envelopes. The minimum cost at a specified level and the minimum cost at a specified slope may be attained by different fitted stories. Neither scalar envelope records the minimum cost subject to both values simultaneously. More generally, preserving several quantities together requires conditioning on their joint value.

Throughout this section, $\theta$ denotes the complete continuous parameter vector of a proposed segment, while $\xi$ denotes the scalar or joint vector exposed by a particular conditional envelope. We write $F_t(\xi)$ for the minimum accumulated cost through $t$ among compatible fitted stories supplying that specified value, with all remaining continuous parameters minimized out. The choice $\xi=\varnothing$ denotes an unconditional minimum. Separate collections are retained whenever the current segment form or other discrete information changes which continuations are admissible; these discrete indices are suppressed below.

The conditional dimension is determined by the independently varying quantities retained across a boundary. A segment may contain many coefficients while a particular transition inherits only one of them, or several affine combinations of them. Conversely, a shared seasonal component may require a large conditional vector even when the local trend is simple.

Every generated candidate must be profiled into all conditional envelopes required by its permitted continuations before any of those profiles is pruned. A candidate that contributes nowhere to one envelope may still be indispensable to another. This is the same requirement encountered for the separate level and slope envelopes in the linear construction.

\subsection{The general recursion and its exactness}

Consider a proposed segment on $(r,t]$ with parameter vector $\theta$ and a permitted transition $d$ at $r$. Once $\theta$ is specified, the transition determines the value $\xi_r$ that the preceding fit must supply. Evaluating the appropriate compatible envelope through $r$ at that value gives the extension cost
\begin{equation}
F_r(\xi_r)+C_{r,t}(\theta)+\beta_d.
\label{eq:general-extension}
\end{equation}
Here $F_r$ denotes the envelope required by $d$: it may condition on level, slope, a joint vector, or no continuous quantity. The notation $C_{r,t}(\theta)$ suppresses the segment-form superscript in \eqref{eq:typed-objective}.

Now fix the conditional envelope to be constructed at $t$, and let $\xi_t$ denote its exposed value. The roles of $\xi_r$ and $\xi_t$ are distinct. The transition at $r$ determines what the outgoing segment inherits; the chosen envelope at $t$ determines what remains exposed for a possible later transition. These vectors need not contain the same quantities or have the same dimension. Both are determined from the proposed segment parameters, and holding $\xi_t$ fixed restricts the minimization to parameter values producing that specified outgoing value.

The general recursion is
\begin{equation}
F_t(\xi_t)
=
\min
\left\{
\begin{aligned}
&\min_{\substack{\text{admissible initial segment}\\
                 \xi_t\ \mathrm{fixed}}}
\left\{
C_{0,t}(\theta)+\beta_{\mathrm{initial}}
\right\},\\
&\min_{\substack{0<r<t\\ d\ \mathrm{permitted}}}
\;
\min_{\substack{\text{admissible outgoing segment}\\
                 \xi_t\ \mathrm{fixed}}}
\left\{
F_r(\xi_r)+C_{r,t}(\theta)+\beta_d
\right\}.
\end{aligned}
\right.
\label{eq:general-recursion}
\end{equation}
The segment minimizations include the choice of permitted segment form and its continuous parameters. They impose the selected boundary relationships, any fixed outgoing quantities, minimum segment lengths, and all other admissibility conditions. The incoming envelope is restricted to compatible discrete states, and an infeasible minimization has value $+\infty$.

Equation~\eqref{eq:general-recursion} is evaluated for every required conditional envelope and compatible discrete state. It adds the earlier conditional cost to the outgoing segment loss before minimizing over quantities not exposed at $t$, allowing observations on both sides of a boundary to estimate their shared parameters jointly. At the terminal endpoint no continuous quantity needs to remain exposed, so the final objective is the minimum over all admissible terminal costs.

The earlier recursions are special cases. Taking $\xi_r=\xi_t=\varnothing$ gives ordinary optimal partitioning. For piecewise-constant segments separated by level resets, taking $\xi_r=\varnothing$ and $\xi_t=\phi$ gives the level-conditioned cost underlying FPOP. For linear segments joined continuously, taking $\xi_r=\psi$ and $\xi_t=\phi$ gives the CPOP recursion in \eqref{eq:cpop-recursion}. Allowing transitions to inherit level, slope, or neither gives the first-order GRACE recursion in \eqref{eq:linear-update}--\eqref{eq:constant-update}. Joint preservation uses the corresponding higher-dimensional envelopes without changing the recursion.

\begin{theorem}[Exactness of the general recursion]
\label{thm:exact}
Suppose that segment and transition costs are additive. For each proposed outgoing segment, suppose that its inherited vector, together with the discrete information retained by the recursion, contains all information from the preceding fit needed to determine the admissibility and cost of that segment and any later continuation. If \eqref{eq:general-recursion} considers every admissible initial segment, previous boundary, outgoing segment form, and permitted transition, retains the joint conditional envelope for every vector a later transition may require, and performs every required conditional minimization globally with its minimum attained whenever feasible, then it attains the global minimum of \eqref{eq:typed-objective}.
\end{theorem}

\begin{proof}
Take any admissible story whose final segment begins at $r$ through transition $d$. Once the outgoing segment is fixed, its dependence on the preceding fit is completely specified by the inherited value $\xi_r$ and the compatible discrete state. Replacing that preceding fit by a minimizer represented in $F_r(\xi_r)$ cannot increase the cost and leaves the outgoing segment and every later continuation admissible. The minimum cost of this extension is therefore \eqref{eq:general-extension}.

The recursion considers every admissible way to make such an extension and minimizes over the outgoing parameters subject to the specified value of $\xi_t$. Conversely, every feasible term in the recursion joins a compatible preceding fit to an admissible outgoing segment through a permitted transition. Starting from the initial-segment terms, induction over $t$ establishes that each envelope equals the minimum cost among all compatible stories supplying its exposed value. The unconditional terminal minimum consequently gives the global minimum of \eqref{eq:typed-objective}.
\end{proof}

Algorithm~\ref{alg:general-grace} gives the corresponding order of computation. Each retained incoming candidate is extended through every compatible transition, and every generated candidate is profiled into all conditional envelopes required by its permitted continuations before pruning is applied.

\begin{algorithm}[H]
\caption{General GRACE recursion}
\label{alg:general-grace}
\KwIn{Observations; permitted segment models and transition types; penalties; admissibility conditions.}
Identify every scalar or joint vector that a permitted transition may inherit\;
Create the corresponding conditional envelopes, separated by any required discrete states, together with the necessary unconditional costs\;
\For{$t=1,\ldots,n$}{
  Generate every admissible initial-segment candidate, including its initial penalty\;
  \For{each admissible $r<t$, transition $d$, and outgoing segment form}{
    Select the compatible incoming envelope required by $d$\;
    Determine the inherited value $\xi_r$ required by transition $d$ from the outgoing segment parameters $\theta$\;
    \For{each retained candidate $f_r$ represented in that envelope}{
      Form $f_r(\xi_r)+C_{r,t}(\theta)+\beta_d$\;
      Retain the generated candidate with its segment and boundary constraints if feasible\;
    }
  }
  \For{each generated candidate}{
    \For{each compatible conditional envelope required at $t$}{
      Hold $\xi_t$ fixed and minimize over the remaining continuous parameters subject to the candidate's constraints\;
      Add the resulting profile to that envelope and store the choices needed for traceback\;
    }
    Minimize over all continuous parameters when an unconditional cost is required\;
  }
  Apply valid pruning rules separately to the completed envelopes\;
}
Select the minimum admissible unconditional cost at $n$ and reconstruct its story and fitted parameters\;
\KwOut{The globally minimizing story, changepoint locations and types, and jointly fitted parameters.}
\end{algorithm}

The recursion and Algorithm~\ref{alg:general-grace} do not require quadratic losses or affine boundary maps. With more general losses or boundary relationships, the construction remains exact under Theorem~\ref{thm:exact}, although its conditional minimizations and envelope calculations may require different methods. We next identify a broad class of models for which every candidate extension and conditional minimization preserves a quadratic representation.

\subsection{Affine updates and quadratic closure}

For the quadratic construction, suppose that the segment loss is convex quadratic in $\theta$ and that the boundary quantities are affine functions of $\theta$. In particular, the value required from the preceding fit has the form
\begin{equation}
\xi_r=A_{r,t}\theta+b_{r,t},
\label{eq:affine-update}
\end{equation}
where the matrix and offset depend on the proposed segment form and transition; these indices are suppressed. The chosen exposed value $\xi_t$ is likewise affine in $\theta$. Quantities fixed by a transition and other continuous feasibility restrictions are imposed through affine equalities.

These conditions cover least-squares models whose fitted values are linear in their unknown coefficients. The basis functions may include polynomials, observed covariates, fixed-knot spline bases, seasonal indicators, Fourier terms at prescribed frequencies, or exponentials with fixed rates. Such components may also be combined. Their dependence on time may be nonlinear while the segment loss remains quadratic in the unknown coefficients.

Polynomial boundary updates illustrate the construction directly. Write a cubic segment in endpoint coordinates as
\[
m_\theta(i)
=
\phi+\dot\phi(i-t)
+\frac{1}{2}\ddot\phi(i-t)^2
+\frac{1}{6}\phi^{(3)}(i-t)^3,
\qquad
\theta=(\phi,\dot\phi,\ddot\phi,\phi^{(3)})^\top.
\]
For $h=t-r$, a transition imposing continuity of value, first derivative, and second derivative requires
\[
\xi_r
=
\begin{pmatrix}
m_\theta(r)\\
m_\theta'(r)\\
m_\theta''(r)
\end{pmatrix}
=
\begin{pmatrix}
1 & -h & h^2/2 & -h^3/6\\
0 & 1 & -h & h^2/2\\
0 & 0 & 1 & -h
\end{pmatrix}
\theta.
\]
Thus a cubic segment has four coefficients, while this transition inherits three quantities jointly. Fixing the highest derivatives at zero gives lower-order pieces, and selecting different boundary quantities gives different permitted transitions. Shared regression or seasonal coefficients enter the inherited vector in the same way.

\begin{proposition}[Quadratic closure]
\label{prop:quadratic-closure}
Suppose that segment losses are convex quadratic, all continuous feasibility restrictions for a fixed story are affine equalities, and exposed quantities are affine functions of the segment parameters. If each conditional minimization has a finite minimum on its attainable domain, every candidate profile generated by the recursion is convex quadratic on an affine attainable domain, with value $+\infty$ outside that domain. The result permits singular quadratic Hessians and nonunique conditional minimizers.
\end{proposition}

\begin{proof}
Affine substitution and addition preserve convex quadratic form. After imposing the affine equalities, write the remaining cost at an attainable exposed value $\xi$ as
\[
z^\top H z+2z^\top\ell(\xi)+q(\xi),
\]
where $z$ contains the free coordinates, $H$ is positive semidefinite, $\ell$ is affine, and $q$ is quadratic. Finiteness implies $\ell(\xi)\in\operatorname{range}(H)$, so minimizing over $z$ gives
\[
q(\xi)-\ell(\xi)^\top H^\dagger\ell(\xi),
\]
where $H^\dagger$ is the Moore--Penrose pseudoinverse. This expression is quadratic, and partial minimization preserves convexity. The attainable domain is affine because it is the affine image of an affine feasible set. Induction over candidate extensions proves the claim.
\end{proof}

The proposition applies to each candidate profile. The lower envelope of those profiles need not itself be convex or quadratic.

Theorem~\ref{thm:exact} establishes exact optimization of the declared objective, while Proposition~\ref{prop:quadratic-closure} provides a closed representation of its candidate costs in the quadratic setting. Neither result guarantees polynomial computational complexity or practical tractability for every transition dictionary. Computational cost depends on the dimensions of the required envelopes and the numbers of surviving candidates. Conservative pruning preserves exactness by retaining candidates whose deletion cannot be certified. 

We next examine the computational consequences of
higher-dimensional conditional envelopes before illustrating
the general framework with polynomial trends and a seasonal
component shared across regimes.

\section{Pruning and computation in higher dimensions}
\label{sec:computation}

The three pruning arguments developed in Section~\ref{sec:linear-pruning} continue to apply when an envelope retains several quantities jointly. Functional pruning asks whether a candidate attains the minimum anywhere on the relevant conditional domain. The $k\beta$ rule compares one permitted continuation with another that can produce the same outgoing segment and preserve every relationship required thereafter. Objective-bound pruning asks whether any admissible completion of the candidate can improve the best complete story already found.

Their mathematical justification is unchanged, but the geometry becomes more difficult as the number of retained quantities increases. For a positive-definite quadratic conditional cost, the values that can still improve the incumbent lie in the ellipsoid defined by \eqref{eq:ellipsoid}. If this ellipsoid does not intersect the candidate's attainable domain, the candidate may be discarded. Otherwise, functional pruning requires determining whether the candidate is cheapest anywhere within the intersection. Positive-semidefinite candidates may instead have unbounded feasible directions and must retain the corresponding feasibility conditions.

Each comparison between two quadratic candidates divides the conditional space by a quadratic surface. To prove that a candidate never attains the envelope, one may therefore need to show that its attainable domain, its objective-feasible ellipsoid, and all regions on which it defeats the competing candidates have no common point. In one dimension these sets reduce to intervals and their intersections can be calculated directly. In higher dimensions the regions may have curved boundaries, may fail to be convex, and may intersect in complicated ways.

An implementation may consequently apply these pruning rules conservatively. A candidate is removed only when its irrelevance can be certified over its entire remaining domain. Failure to establish deletion leaves an unnecessary candidate in the recursion and increases computation, but it cannot change the selected optimum. The difficulty lies in proving that the relevant quadratic regions have empty intersection, rather than in the validity of the pruning arguments.

Every comparison must also retain the information required by the continuation being considered. Candidates can be compared within an envelope only when they condition on the same quantities. For example, a candidate that has already minimized over the seasonal coefficients cannot replace one whose future continuation must inherit those coefficients. Similarly, $k\beta$ pruning requires an alternative transition that can produce the same outgoing segment while satisfying the same boundary relationships, segment-length requirements, and other restrictions. If a comparison eliminates only one kind of continuation, the candidate is removed from that use and remains available wherever the comparison does not apply.

\subsection{Obtaining a useful upper bound}

Any admissible complete story provides an upper bound $U$ on the optimal objective. We obtain such a story using a profiled version of GRACE, called \emph{First GRACE}, in which every candidate is profiled over endpoint level. First GRACE is a separate, deliberately reduced preliminary search. The exact first-order recursion of Section~\ref{sec:linear-grace} retains both level and slope, whereas First GRACE retains only level and is used solely to propose a complete story.

At each endpoint $t$, First GRACE considers the permitted segment forms and transition types but conditions only on endpoint level. For every generated candidate, all continuous quantities other than $\phi$---including slope, curvature, and seasonal coefficients---are minimized at fixed $\phi$. When a new segment is attached, the boundary relationship is imposed using the conditional minimizers stored with the retained level candidate, and the resulting extension is profiled back over endpoint level. Separate level-conditioned collections are retained when the form of the current segment changes which transitions may follow, just as the first-order recursion distinguishes $F_{t,\mathrm{lin}}^\phi$ from $F_{t,\mathrm{const}}^\phi$. First GRACE therefore propagates a single conditional coordinate, endpoint level, rather than the several scalar and higher-dimensional coordinates required by the exact recursion.

The proposed story is obtained by minimizing over endpoint level and over all admissible terminal collections.
Because First GRACE has profiled out quantities that a later transition may need to inherit, this story need not minimize the full objective. A candidate that never attains the level envelope may still have been optimal conditional on a particular slope, curvature, seasonal vector, or joint collection of quantities.

The preliminary objective is therefore not itself used as the upper bound. Instead, the changepoint locations, transition types, and segment forms selected by First GRACE are held fixed, and all continuous parameters are refitted jointly under the original boundary relationships. If this refit is admissible, its evaluated objective supplies a valid upper bound $U$. First GRACE is used only to select a promising complete story; the bound is rigorous because it is calculated after that story has been refitted under the full model.

\subsection{Lower bounds on the remaining cost}

A lower bound $L_t$ on the additional cost after time $t$ can be obtained by relaxing the continuation problem. Future segments may be fitted independently, preservation requirements across future boundaries may be removed, and future penalties may be omitted or reduced. The resulting optimum is a valid lower bound whenever every admissible continuation after $t$ is represented in the relaxed problem at no greater cost. The choice $L_t=0$ already preserves exactness when all remaining losses and penalties are nonnegative; stronger bounds are optional accelerators.

Such a relaxation can often be solved by an ordinary segmentation dynamic program working backward from $n$. If the continuations available after $t$ depend on the regime reaching that time, separate lower bounds may be calculated for the relevant regimes. A tighter value of $L_t$ can eliminate more candidates through \eqref{eq:feasible}, although computing it may require a more expensive preliminary calculation.

This construction provides a general route to valid lower bounds rather than a definitive solution for every transition dictionary. Developing bounds that are simultaneously tight, inexpensive to compute, and sensitive to the permitted boundary relationships remains an active area of research. The framework above identifies the condition required for validity and indicates how such bounds enter the exact recursion; designing the most effective relaxation for particular model classes is left for future work.
\subsection{Parallelization}

The recursion proceeds sequentially over endpoints $t=1,\ldots,n$, since the envelopes at time $t$ depend on those constructed at earlier endpoints. Once the earlier envelopes are available, most calculations within a fixed endpoint are independent.

Extensions from different previous boundaries $r<t$ can be generated independently. For each $r$, the permitted transition types can be considered independently, and each retained incoming quadratic can be extended independently within a transition type. Every resulting candidate can then be profiled independently into each conditional envelope required by a possible later transition. The principal inner calculation therefore decomposes over
\[
\{\text{previous boundary}\}
\times
\{\text{transition type}\}
\times
\{\text{incoming candidate}\}
\times
\{\text{required envelope}\}.
\]
Functional, $k\beta$, and objective-bound comparisons can likewise be distributed across candidates and envelopes before their results are combined.

GRACE therefore has a sequential outer loop over observation indices and a broad collection of independent calculations within each iteration. The available parallel work grows with the number of retained candidates, so the endpoints that are most expensive serially also offer the greatest opportunity for acceleration. Segment-cost calculations, searches for admissible upper bounds, lower bounds on the remaining cost, and searches over fixed tuning-parameter values provide additional independent work. This organization admits multicore, accelerator, and distributed implementations, while the cost of certifying intersections among higher-dimensional quadratic regions remains a central computational consideration.

\section{Example: Polynomial trends with persistent seasonality}
\label{sec:grace-seasonal-example}

Many time series contain a recurring seasonal pattern together with a trend
whose level, slope, or curvature changes at isolated times. The seasonal
pattern can be estimated from the entire series while the trend and its
boundary relationships change locally. For a known period $P$, consider
\begin{equation}
y_i=T_j(i)+\pi_{\mathrm{phase}(i)}+\epsilon_i,
\qquad \tau_{j-1}<i\leq\tau_j,
\label{eq:seasonal-model}
\end{equation}
where $T_j$ is constant, linear, or quadratic on regime $j$, and
$\pi=(\pi_1,\ldots,\pi_P)$ is shared across regimes. The constraint
$\sum_{p=1}^{P}\pi_p=0$ separates seasonality from the trend intercept,
leaving $P-1$ free seasonal coefficients.

For a quadratic trend parameterized at the right endpoint $t$, the segment
loss is
\begin{equation}
C_{r,t}(\phi,\dot\phi,\ddot\phi,\pi)
=\sum_{i=r+1}^{t}\left\{
y_i-\phi-\dot\phi(i-t)-\tfrac12\ddot\phi(i-t)^2
-\pi_{\mathrm{phase}(i)}\right\}^{2}.
\label{eq:seasonal-cost}
\end{equation}
A linear regime fixes $\ddot\phi=0$, and a constant regime also fixes
$\dot\phi=0$. After imposing the zero-sum constraint, the loss remains
quadratic in the trend and seasonal coefficients. At a changepoint, $\pi$
must remain exposed, together with whichever trend quantities the transition
preserves. These transitions have the affine form \eqref{eq:affine-update}, so the
dynamic-programming identity in Theorem~\ref{thm:exact} continues to apply. Table~\ref{tab:seasonal-types} gives representative examples. 

\begin{table}[ht]
\centering
\small
\begin{tabular}{p{0.47\linewidth}p{0.40\linewidth}}
\toprule
Trend transition & Quantities retained jointly\\
\midrule
Level shift preserving slope and curvature & $\dot\phi,\ddot\phi,\pi$\\
Slope change preserving level and curvature & $\phi,\ddot\phi,\pi$\\
Curvature change preserving level and slope & $\phi,\dot\phi,\pi$\\
Trend termination & $\phi,\pi$; outgoing derivatives fixed at zero\\
Trend resumption & $\phi,\pi$\\
Complete trend reset & $\pi$\\
\bottomrule
\end{tabular}
\caption{Representative trend transitions with one seasonal pattern shared
across all regimes.}
\label{tab:seasonal-types}
\end{table}

The computational difficulty is substantial, however. With $P=12$, even a
complete trend reset retains eleven seasonal coordinates; preserving two
trend quantities requires a thirteen-dimensional conditional envelope.
Functional pruning must establish dominance throughout the relevant domain.
The objective-based feasible ellipsoids obtained from upper and lower bounds
reduce that domain, but can still leave too many competing quadratics in
these dimensions. The present pruning machinery is insufficient to make
the unrestricted persistent-seasonality problem routinely tractable.

We therefore first use an alternating procedure to obtain a well-supported
estimate of seasonality. Starting from a robust seasonal estimate based on
median differences between consecutive observations, grouped by phase, we
deseason the series and use \emph{First GRACE} to propose a trend structure.
We jointly refit its trend and seasonal coefficients by least squares, then
search within five observations of each proposed changepoint and over the
compatible transition types. Improving updates are accepted under the
original penalized objective. The updated seasonal estimate is used to
deseason the series again, and the procedure is repeated. Local refinement
holds the current regime count and degrees fixed; a new profiled proposal
can change both.

The preliminary fit serves to construct an informative,
data-dependent region for the shared seasonal coefficients.
Its proposed changepoint configuration is not imposed on the
subsequent exact optimization: the number, locations, polynomial
degrees, and types of regimes remain free to vary within the
specified seasonal region.


We construct the seasonal region from the best preliminary fit. Write
$\pi=B\eta$, where the columns of $B$ form an orthonormal basis for the
zero-sum subspace, and let $q=P-1$. Joint OLS gives $\widehat\eta$ and its
covariance block $\widehat V_\eta$, accounting for estimation of the continuous
trend coefficients. With $\nu$ residual degrees of freedom, we use
\begin{equation}
\mathcal E=\left\{B\eta:
(\eta-\widehat\eta)^{\mathsf T}\widehat V_\eta^{-1}
(\eta-\widehat\eta)\leq qF_{q,\nu}(0.95)\right\}.
\label{eq:seasonal-region}
\end{equation}
This is a nominal $95\%$ OLS ellipsoid: its calibration treats the selected
trend structure as fixed and does not account for changepoint selection.
Its center and size are estimated from the data, without using the
generating trend or seasonal coefficients.

We then hold $\mathcal E$ fixed and run the full GRACE recursion subject to
$\pi\in\mathcal E$. A new profiled proposal and constrained local refinement
provide a feasible complete fit and hence an upper bound. Valid lower
bounds and improved feasible fits can be refreshed during the recursion.
The seasonal restriction reduces the domain over which pruning must certify
dominance, while the exact search remains free to change the number,
locations, degrees, and types of regimes. Both the trend and seasonality
are re-estimated. Unlike an objective-derived pruning region, the imposed
OLS ellipsoid changes the optimization problem: the resulting guarantee is
global optimality within this fixed region.

Figure~\ref{fig:overview} shows a completed exact GRACE fit constrained to $\pi\in\mathcal E$ with minimum segment length $24$, seasonal period $12$, and all 30
legal transitions among constant, linear, and quadratic regimes. The
generating regimes have length $36$; the shorter minimum permits up to
nine regimes in the 216 observations. GRACE recovers all six generating
regimes and their transition types. Four boundaries are recovered exactly,
and the boundary at 108 is estimated at 109. The shared seasonal pattern
is also recovered closely. An informative seasonal region therefore makes
this joint search tractable with a minimum length substantially below the
generating regime length. The example demonstrates that heterogeneous polynomial
transitions and a persistent seasonal component can be estimated
jointly, with global optimality guaranteed over the specified
seasonal region.

\begin{figure}[!htbp]
\centering
{\large\bfseries
Joint estimation of polynomial regimes
and persistent seasonality\par}
\vspace{0.5em}
\paperfigure{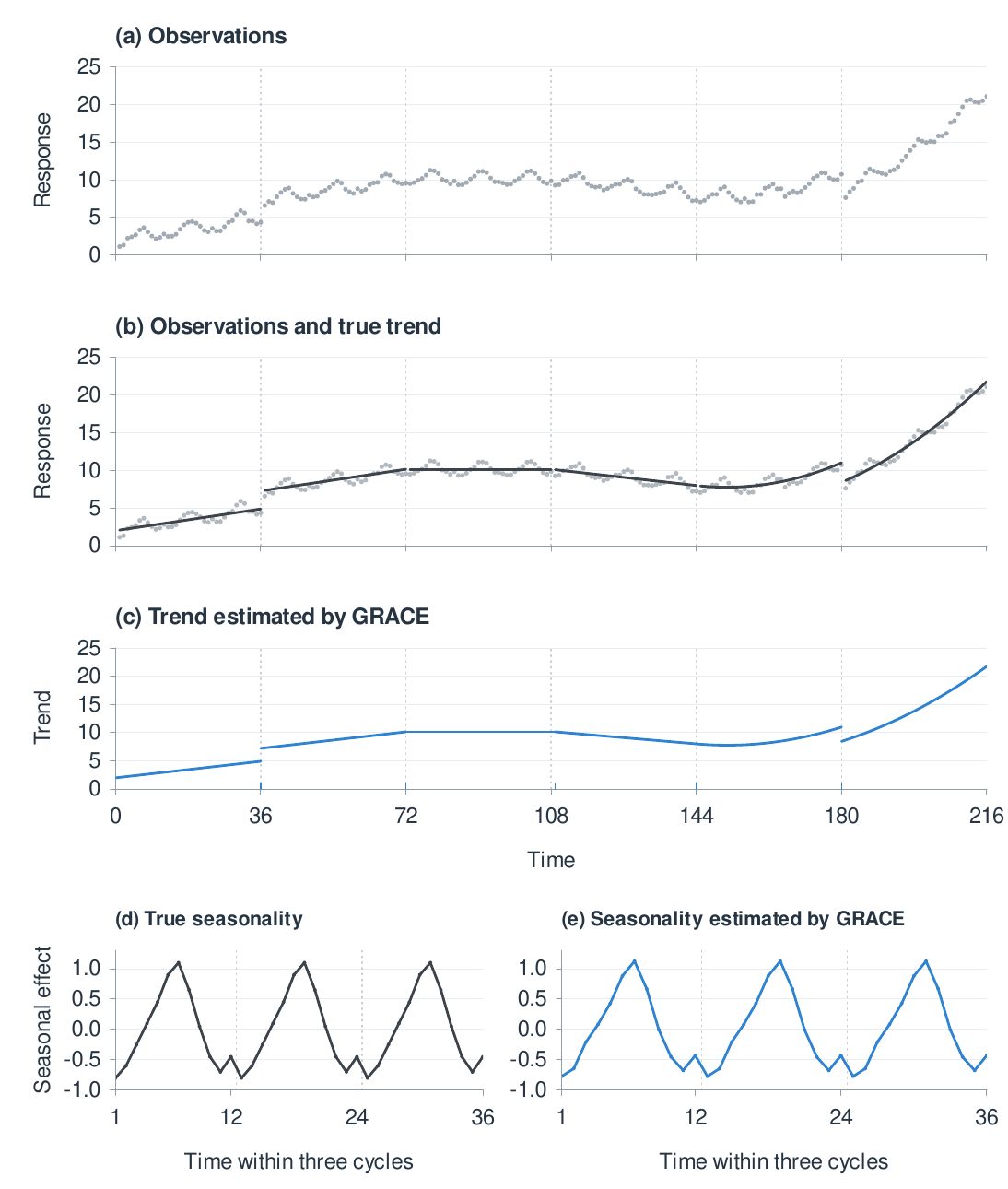}{6.0in}
\caption{A changing polynomial trend with persistent seasonality, fitted
with minimum segment length 24. (a) Observations. (b) Observations and the
true trend. (c) Trend estimated by GRACE. (d) True seasonality and
(e) seasonality estimated by GRACE, each repeated over three periods on
identical axes. Truth is charcoal and estimates obtained by GRACE are blue.
Light dotted vertical lines in (a)--(c) mark generating changepoints; blue
ticks in (c) mark estimated changepoints. Dotted separators in (d) and (e)
mark cycles. The data contain 216 observations, six generating regimes of
length 36, and independent Gaussian noise with standard deviation 0.15.
The selected transitions are a level shift preserving slope, trend
termination, trend resumption, the introduction of curvature, and a level
shift preserving slope and curvature. The fit globally minimizes the
penalized objective over all permitted structures within the fixed nominal
$95\%$ OLS seasonal ellipsoid.}
\label{fig:overview}
\end{figure}
\FloatBarrier

\section{Discussion and extensions}
\label{sec:discussion}

This paper has given an exact algorithm for a broad family of changepoint model-selection problems. Within a finite collection of segment models and permitted boundary relationships, GRACE jointly selects the number and locations of changepoints, the type of change occurring at each boundary, the model fitted between successive boundaries, and all continuous parameters. Whenever the loss is additive and every transition depends on the fit through the preceding boundary only through a finite-dimensional collection of retained quantities, the general recursion attains the global minimum of \eqref{eq:typed-objective}.

CPOP is the continuous piecewise-linear instance of this principle. Its level-conditioned envelope retains exactly the quantity needed to join two linear segments continuously. We have shown that the same envelope also supports trend termination, trend resumption, and complete resets. A discontinuous level shift preserving slope requires an additional slope-conditioned envelope because optimality over endpoint level does not imply optimality over slope. Profiling every generated candidate into both envelopes before pruning gives an exact recursion in which all of these first-order transitions may compete.

The same construction extends beyond lines. A transition may preserve level, slope, curvature, seasonal coefficients, or any joint collection of quantities through which the outgoing segment depends on the fit through the boundary. GRACE retains the accumulated cost conditional on precisely that collection. When several quantities are preserved together, they remain exposed together in a higher-dimensional envelope; separate marginal envelopes do not contain the required joint conditional cost. This gives an exact calculus for changes in polynomial order, discontinuities in selected derivatives, persistent seasonality, and other models connected by finite-dimensional affine updates.

The transition types define the scientific question being asked. Once those possibilities have been declared, GRACE searches over their locations and combinations without committing greedily to a changepoint or to the kind of change occurring there. The exactness theorem therefore covers an entire family of objectives rather than a single predetermined form of piecewise signal. The linear and seasonal examples illustrate particular members of that family, while the general recursion specifies how additional segment models and boundary relationships enter the same optimization.

The principal computational cost is determined by the dimension of the conditional envelopes and the number of candidates that remain capable of participating in an optimum. Functional, $k\beta$, and objective-bound pruning reduce this collection without changing the result. First GRACE provides a lower-dimensional level-profiled fit from which a rigorous upper bound can be obtained after joint refitting, while relaxed continuation problems provide lower bounds on the remaining cost. Within each sequential endpoint, candidate generation, profiling, and many pruning comparisons decompose into independent calculations. These features provide a practical route from the general exact recursion to implementations for increasingly rich models.

GRACE builds on functional pruning for optimal partitioning and continuous piecewise-linear fitting \cite{psn,fpop,cpop}. Graph-constrained methods describe permitted relationships among discrete states, and gfpop combines such a graph with a scalar continuous parameter \cite{gfpop}. SPOP studies continuously differentiable quadratic splines with a fixed number of segments and obtains a practical algorithm by locally discretizing admissible boundary values and initial derivatives \cite{spop}. Geometric functional pruning and DUST address multivariate segment parameters when successive segments are otherwise independent \cite{geom,dust}. GRACE addresses a distinct dependence structure in which different transitions may preserve different projections of the incoming parameters and some quantities must remain jointly estimated across changes in regime.

\subsection{Likelihood-based models and dependent errors}
\label{sec:likelihood}

The general recursion also applies beyond least squares. Let $C_{r,t}(\theta)$ be an additive negative log likelihood or another segment loss. If the quantities retained at a boundary contain all information through which the fit up to that boundary can affect the outgoing segment, then the conditional replacement argument underlying Theorem~\ref{thm:exact} remains valid. Exact optimization requires the corresponding conditional minimizations to be performed globally.

Quadratic least squares is especially convenient because affine substitution and profiling preserve quadratic candidate costs. This gives closed-form updates and permits the pruning calculations developed above. Other likelihoods may produce conditional costs that are piecewise convex, nonconvex, or represented by functions other than quadratics. The same recursion still specifies the required minimization, although the representation of the conditional costs and the calculations used to compare them must be adapted to the chosen likelihood.

\paragraph{AR(1) residual dependence.}

Suppose
\[
e_i=y_i-f_i
\]
follows a Gaussian AR(1) process with coefficient $\rho$. Conditional on the initial residual, the innovation contribution is
\begin{equation}
\sum_{i=2}^{n}(e_i-\rho e_{i-1})^2.
\label{eq:ar1-loss}
\end{equation}
Equation~\eqref{eq:ar1-loss} is the conditional innovation contribution; the treatment of the initial residual must be included in the initial-segment cost.
The innovation crossing a changepoint depends on fitted values from both sides of the boundary. An exact recursion must therefore retain, or recover from its retained quantities, the fitted value needed to evaluate the incoming residual. A structural reset of the fitted mean does not remove this dependence unless the error process is also declared to restart at that boundary.

For fixed $\rho$, \eqref{eq:ar1-loss} remains quadratic in the fitted-mean parameters. The required boundary value can therefore be added to the quantities retained by each transition, after which the same quadratic recursion applies. If $\rho$ is selected from a finite collection of values, the exact recursion may be run separately for each value and the best resulting objective selected. Continuous estimation of $\rho$ requires either a globally valid outer minimization or a conditional representation that retains the dependence on $\rho$.

Likelihoods with nonquadratic conditional costs, continuously estimated dependence parameters, and more general residual processes provide natural directions for extending the computational machinery. The general recursion already identifies the information that must cross each boundary; the remaining task is to develop representations and pruning calculations suited to those conditional costs.

The central contribution can therefore be stated succinctly. A changepoint determines which features of the fit through the boundary remain binding on the segment that follows. Once those quantities are identified and retained jointly, the changepoint locations, transition types, segment models, and continuous parameters can be selected within one exact dynamic program.

\appendix
\section{Penalty convention and the Schwarz criterion}
\label{app:penalty-convention}

For independent Gaussian errors with common variance $\sigma^2$ held fixed, the maximized log likelihood satisfies
\[
-2\ell(\widehat\theta) = \frac{\mathrm{RSS}}{\sigma^2}+\text{constant}.
\]
The Schwarz criterion adds $q\log n$ for a regular model with $q$ estimated parameters \cite{schwarz}. Multiplying by $\sigma^2$ expresses this criterion on the residual-sum-of-squares scale:
\[
\mathrm{RSS}+q\sigma^2\log n.
\]
This gives the parameter-penalty unit $\beta=\sigma^2\log n$ used in the main text. When the variance is estimated beforehand, its supplied estimate is held fixed throughout the optimization.

Let $p_0$ be the number of independent continuous parameters in the initial regime, and let $p_{d_j}$ be the number introduced by transition $d_j$. For the preservation, release, and fixing operations considered here, the continuous-parameter count is
\[
q_{\mathrm{continuous}} = p_0+\sum_{j=1}^{m}p_{d_j}.
\]
A preserved parameter is counted only when first introduced. A released parameter is estimated afresh and contributes a new parameter. Fixing an outgoing parameter at a specified value contributes none.

We additionally assign one penalty unit to each estimated changepoint location. The total complexity penalty is therefore
\[
\beta\left(m+p_0+\sum_{j=1}^{m}p_{d_j}\right)
=
p_0\beta+\sum_{j=1}^{m}(1+p_{d_j})\beta,
\]
giving the initial-model penalty $p_0\beta$ and transition penalty $\beta_d=(1+p_d)\beta$.

Consequently, a continuous slope change or trend resumption costs $2\beta$ for location and slope; a level shift preserving slope costs $2\beta$ for location and level; a reset to a line costs $3\beta$ for location, level, and slope; trend termination costs $\beta$ for location alone; and a reset to a constant costs $2\beta$ for location and level.

The one-unit location charge is an explicit extension of the Schwarz parameter-counting convention. Its statistical calibration requires a changepoint-specific argument because locations are discrete and nonregular. The exactness of the GRACE recursion concerns optimization of the prescribed criterion and is unaffected by how these penalties are statistically calibrated.

\section{Explicit level-to-slope coefficient transformation}
\label{app:transformation}

This appendix gives the coefficient map for a linear candidate generated from an incoming endpoint quadratic and the observations on $(r,t]$. All observation losses are ordinary residual sums of squares.

\subsection{Mapping the ending-level quadratic}

Suppose the candidate's ending-level profile is
\[
f_t^\phi(\phi)=A\phi^2+B\phi+C,
\]
and its generating cost is
\[
f_r(\psi)+C_{r,t}(\psi,\phi)+\beta_d,
\]
where $f_r$ is an incoming endpoint quadratic and $\beta_d$ is the penalty for the transition that generated the candidate.

Define the following quantities from the generating segment:
\begin{equation}
\begin{split}
V &= \sum_{i=r+1}^{t}\left(1-\frac{i-r}{t-r}\right)\frac{i-r}{t-r}
   = \frac{(t-r)^2-1}{6(t-r)},\\
W &= \sum_{i=r+1}^{t}\left(\frac{i-r}{t-r}\right)^2
   = \frac{(t-r+1)\{2(t-r)+1\}}{6(t-r)},\\
R &= \sum_{i=r+1}^{t}\frac{i-r}{t-r}\,y_i.
\end{split}
\label{eq:map-data}
\end{equation}
Set
\[
G=(V+W)^2-A(2V+W).
\]
When $G>0$, its slope-conditioned quadratic is
\[
f_t^{\dot\phi}(\dot\phi)=a\dot\phi^2+b\dot\phi+c,
\]
where
\begin{equation}
\boxed{
\begin{aligned}
a &= \frac{(t-r)^2AV^2}{G},\\[2mm]
b &= \frac{(t-r)V\{(V+W)B+2AR\}}{G},\\[2mm]
c &= C+\frac{(2V+W)B^2+4(V+W)BR+4AR^2}{4G}.
\end{aligned}
}
\label{eq:explicit-map}
\end{equation}

The map uses $A,B,C$ and the candidate's generating segment data. It is applied candidate by candidate before the level and slope envelopes select their respective winners.

\subsection{Derivation}

Expanding the generating quadratic gives
\[
h\psi^2+2V\psi\phi+W\phi^2+\ell\psi-2R\phi+e,
\]
where $h,\ell,e$ include the incoming candidate cost and the remaining segment terms.

Minimizing over $\psi$ gives
\[
A=W-\frac{V^2}{h},
\qquad
B=-2R-\frac{V\ell}{h},
\qquad
C=e-\frac{\ell^2}{4h}.
\]
Substituting $\psi=\phi-(t-r)\dot\phi$ makes the coefficient of $\phi^2$ equal to $h+2V+W$. Minimizing over $\phi$ therefore gives
\[
\begin{aligned}
a &= (t-r)^2\left\{h-\frac{(h+V)^2}{h+2V+W}\right\},\\
b &= -(t-r)\ell+\frac{(t-r)(h+V)(\ell-2R)}{h+2V+W},\\
c &= e-\frac{(\ell-2R)^2}{4(h+2V+W)}.
\end{aligned}
\]
Eliminating $h,\ell,e$ using the expressions for $A,B,C$ yields \eqref{eq:explicit-map}.

The direct map applies to the stated generating structure when $G>0$ and there are no additional domain restrictions. When these conditions fail, the slope profile is constructed directly from the generating conditional cost. Candidates arriving through other retained quantities use the corresponding substitution and profiling calculation.

\subsection{A direct formula from the incoming quadratic}

The slope coefficients can also be computed directly while appending the segment. Suppose
\[
f_r(\psi)=A_r\psi^2+B_r\psi+C_r.
\]
Define
\[
\begin{array}{lll}
S_0=t-r,
&
S_1=\displaystyle\sum_{i=r+1}^{t}(i-r),
&
S_2=\displaystyle\sum_{i=r+1}^{t}(i-r)^2,
\\[2mm]
T_0=\displaystyle\sum_{i=r+1}^{t}y_i,
&
T_1=\displaystyle\sum_{i=r+1}^{t}(i-r)y_i,
&
T_2=\displaystyle\sum_{i=r+1}^{t}y_i^2.
\end{array}
\]
The location sums are
\[
S_1=\frac{(t-r)(t-r+1)}{2},
\qquad
S_2=\frac{(t-r)(t-r+1)\{2(t-r)+1\}}{6}.
\]

At fixed slope, the generating cost is
\[
A_r\psi^2+B_r\psi+C_r+\sum_{i=r+1}^{t}\{y_i-\psi-(i-r)\dot\phi\}^2+\beta_d.
\]
Its minimizing starting value is
\[
\psi^*(\dot\phi)=\frac{2T_0-B_r-2S_1\dot\phi}{2(A_r+S_0)}.
\]
Substitution gives
\begin{equation}
\boxed{
\begin{aligned}
a &= S_2-\frac{S_1^2}{A_r+S_0},\\[2mm]
b &= -2T_1-\frac{S_1(B_r-2T_0)}{A_r+S_0},\\[2mm]
c &= C_r+T_2+\beta_d-\frac{(B_r-2T_0)^2}{4(A_r+S_0)}.
\end{aligned}
}
\label{eq:source-map}
\end{equation}
The relation $\psi^*(\dot\phi)$ supplies the starting value needed to reconstruct the selected fit.

\printbibliography
\end{document}